\documentclass[11pt]{amsart}
\usepackage[T1]{fontenc}
\usepackage{lmodern}
\usepackage{microtype}
\usepackage[hidelinks]{hyperref}
\hypersetup{pdftitle={A new upper bound for 7-universal tournaments},
  pdfauthor={Harshit Verma}}

\newtheorem{theorem}{Theorem}
\begin{document}

\title{A new upper bound for 7-universal tournaments}
\author{Harshit Verma}
\address{Department of Computer Science, Yale University, New Haven, CT, USA}
\email{harshit.verma@yale.edu}
\subjclass[2020]{05C20, 05C35}
\keywords{Tournament, universal tournament, induced subtournament}

\begin{abstract}
A tournament is $k$-universal if it contains every tournament on $k$
vertices as a subtournament. We give a $7$-universal tournament on $14$
vertices, improving the upper bound of $15$ due to Zhang and Szeider.
Together with their lower bound, it follows that the smallest $7$-universal
tournament has $13$ or $14$ vertices.
\end{abstract}

\maketitle

\section{Introduction}

A tournament is an orientation of a complete graph. A tournament $T$ is
\emph{$k$-universal} if every tournament on $k$ vertices is isomorphic to a
subtournament of $T$. Let $t(k)$ denote the minimum number of vertices of a
$k$-universal tournament. Universal tournaments were studied by
Moon~\cite[Section~17]{Moon}, who gave general bounds on $t(k)$, and
Alon~\cite[Theorem~4.1]{Alon} proved that $t(k)=(1+o(1))2^{(k-1)/2}$.

Using SAT solvers, Zhang and Szeider~\cite{ZS} showed that $t(k)=1$, $2$,
$4$, $5$, $8$, $10$ for $k=1,\dots,6$, and that $13\le t(7)\le15$. They left
open whether $7$-universal tournaments on $13$ or $14$ vertices
exist~\cite[Section~7.4.2]{ZS}; the same bounds are stated in the
survey~\cite[Section~4.3]{Szeider}. We show that $14$ vertices suffice.

\begin{theorem}\label{thm:main}
There is a $7$-universal tournament on $14$ vertices. Consequently,
$13\le t(7)\le14$.
\end{theorem}

To the best of our knowledge, no $7$-universal tournament on fewer than $15$
vertices was previously known.

\section{The tournament}

Let $T_{14}$ be the tournament on $\{1,\dots,14\}$ whose adjacency matrix
$A$ is given in Table~\ref{tab:matrix}: $A_{ij}=1$ if and only if the edge
between $i$ and $j$ is directed from $i$ to $j$.

\begin{table}[t]
\renewcommand{\arraystretch}{0.9}
\[
A=\begin{pmatrix}
0&1&1&1&0&0&0&1&0&0&1&0&0&0\\
0&0&1&0&0&1&1&1&1&1&0&1&0&1\\
0&0&0&1&0&0&0&1&1&0&1&1&1&0\\
0&1&0&0&0&0&1&1&1&0&0&1&1&1\\
1&1&1&1&0&0&1&1&1&0&1&1&0&1\\
1&0&1&1&1&0&1&0&0&1&1&0&1&0\\
1&0&1&0&0&0&0&1&0&0&0&1&0&0\\
0&0&0&0&0&1&0&0&0&0&1&1&0&1\\
1&0&0&0&0&1&1&1&0&0&1&0&0&1\\
1&0&1&1&1&0&1&1&1&0&0&1&0&1\\
0&1&0&1&0&0&1&0&0&1&0&0&1&1\\
1&0&0&0&0&1&0&0&1&0&1&0&1&0\\
1&1&0&0&1&0&1&1&1&1&0&0&0&0\\
1&0&1&0&0&1&1&0&0&0&0&1&1&0
\end{pmatrix}
\]
\caption{The adjacency matrix of $T_{14}$.}
\label{tab:matrix}
\end{table}

\begin{proof}[Proof of Theorem~\ref{thm:main}]
A tournament on the vertex set $\{1,\dots,7\}$ is determined by the
orientations of its $21$ edges, so there are $2^{21}=2\,097\,152$ such
tournaments, and every tournament on $7$ vertices is isomorphic to one of
them. Each injective map $\varphi\colon\{1,\dots,7\}\to\{1,\dots,14\}$
determines the tournament $T_\varphi$ on $\{1,\dots,7\}$ in which $i\to j$
if and only if $\varphi(i)\to\varphi(j)$ in $T_{14}$, and $T_\varphi$ is
isomorphic to the subtournament of $T_{14}$ induced by the image of
$\varphi$. An exhaustive computation over all
$14\cdot13\cdots8=17\,297\,280$ injective maps shows that every tournament
on $\{1,\dots,7\}$ occurs as some $T_\varphi$. Hence $T_{14}$ is
$7$-universal, and $t(7)\le14$. The lower bound $t(7)\ge13$ is due to
Zhang and Szeider~\cite[Section~7.4.2]{ZS}.
\end{proof}

As an independent check, the $\binom{14}{7}=3432$ subtournaments of $T_{14}$
induced by $7$-element vertex sets fall into $456$ isomorphism classes, the
number of tournaments on $7$ vertices up to isomorphism~\cite[Table~1]{ZS}.
The computation in the proof is carried out by the ancillary file
\texttt{verify.py} on arXiv.

\section{Remarks}

The tournament $T_{14}$ has score sequence
$(4,4,5,5,6,6,6,6,7,7,8,8,9,10)$ and trivial automorphism group. Deleting
any vertex of $T_{14}$ destroys $7$-universality, so $T_{14}$ has no
$7$-universal subtournament on $13$ vertices.

\bibliographystyle{amsplain}
\bibliography{references}

\end{document}